\documentclass[journal]{IEEEtran}
\usepackage{amsmath}
\usepackage{amssymb}
\usepackage{amsfonts}
\usepackage{graphicx}
\usepackage{epsfig}
\usepackage{subfigure}
\usepackage{psfrag}
\usepackage{cite}
\usepackage{latexsym}
\usepackage{url}
\usepackage{color}
\ifCLASSINFOpdf

\else

\fi

\begin{document}
\title{Proactive Eavesdropping via Jamming for Rate Maximization over Rayleigh Fading Channels
\author{Jie Xu,~\IEEEmembership{Member,~IEEE}, Lingjie Duan,~\IEEEmembership{Member,~IEEE}, and Rui Zhang,~\IEEEmembership{Senior~Member,~IEEE}}
\thanks{J. Xu and L. Duan are with the Engineering Systems and Design Pillar, Singapore University of Technology and Design (e-mail:~jiexu.ustc@gmail.com,~lingjie\_duan@sutd.edu.sg).}
\thanks{R. Zhang is with the Department of Electrical and Computer Engineering, National University of Singapore (e-mail: elezhang@nus.edu.sg).}}

\maketitle

\begin{abstract}
Instead of against eavesdropping, this letter proposes a new paradigm in wireless security by studying how a legitimate monitor (e.g., a government agency) efficiently eavesdrops a suspicious wireless communication link. The suspicious transmitter controls its communication rate over Rayleigh fading channels to maintain a target outage probability at the receiver, and the legitimate monitor can successfully eavesdrop only when its achievable rate is no smaller than the suspicious communication rate. We propose a proactive eavesdropping via jamming approach to maximize the average eavesdropping rate, where the legitimate monitor sends jamming signals with optimized power control to moderate the suspicious communication rate.
\end{abstract}

\newtheorem{definition}{\underline{Definition}}[section]
\newtheorem{fact}{Fact}
\newtheorem{assumption}{Assumption}
\newtheorem{theorem}{\underline{Theorem}}[section]
\newtheorem{lemma}{\underline{Lemma}}[section]
\newtheorem{corollary}{\underline{Corollary}}[section]
\newtheorem{proposition}{\underline{Proposition}}[section]
\newtheorem{example}{\underline{Example}}[section]
\newtheorem{remark}{\underline{Remark}}[section]
\newtheorem{algorithm}{\underline{Algorithm}}[section]
\newcommand{\mv}[1]{\mbox{\boldmath{$ #1 $}}}

\section{Introduction}

\IEEEPARstart{W}{ireless} security has attracted a lot of research attentions recently, and there are extensive works investigating defense mechanisms, such as physical-layer security techniques, to secure wireless communications against eavesdropping (see e.g. \cite{ZouWangHanzo2015} and the references therein). However, these existing works usually view eavesdropping as illegitimate attacks, and have limitations from a broader national security perspective. Recently, with technical advancements in e.g. smartphones and drones, criminals or terrorists can potentially use them to establish wireless communication links for committing crimes and terrorism. Therefore, there is a growing need for government agencies (e.g., National Security Agency in the US) to monitor and legitimately eavesdrop these suspicious wireless communications (see, e.g., the Terrorist Surveillance Program \cite{TerroristWiKi}).

In this letter, we propose a paradigm shift in wireless security from preventing conventional eavesdropping attacks to a new legitimate surveillance objective. In particular, we consider a simple setup as shown in Fig. \ref{fig:0}, where a legitimate monitor aims to eavesdrop a point-to-point suspicious wireless communication link over Rayleigh fading channels. We consider that the suspicious transmission is a conventional communication link without using advanced physical-layer security techniques,{\footnote{This consideration is practical since it is difficult for the suspicious transmitter to know the existence of the eavesdropper and obtain the channel state information (CSI) of the eavesdropping channel.}} and the suspicious transmitter controls its communication rate to maintain a target outage probability at the receiver. In this case, the legitimate monitor can successfully decode and eavesdrop the suspicious communication link (in eavesdropping non-outage) only when its  achievable data rate is no smaller than the suspicious communication rate. In practice, the legitimate eavesdropping is very challenging, since the legitimate monitor may be far away from the suspicious transmitter and wireless channels may fluctuate significantly over time.

\begin{figure}
\centering
 \epsfxsize=1\linewidth
    \includegraphics[width=7.5cm]{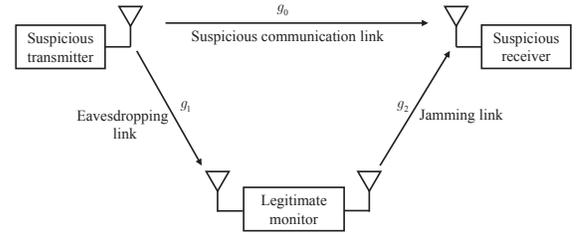}
\caption{A wireless legitimate surveillance scenario, where a full-duplex legitimate monitor aims to eavesdrop a point-to-point suspicious communication link via jamming.} \label{fig:0}\vspace{-1.5em}
\end{figure}

To overcome this issue, we propose an approach named proactive eavesdropping via jamming, where the legitimate monitor, operating in a full-duplex mode, sends jamming signals to moderate the suspicious communication rate for facilitating the simultaneous eavesdropping. In particular, we optimize the jamming power at the legitimate monitor to maximize its average eavesdropping rate, which is defined as the suspicious communication rate multiplied by the eavesdropping non-outage probability. We obtain the optimal jamming power solution in closed-form. Numerical results show that thanks to the optimized power control, the proposed proactive eavesdropping via jamming outperforms the conventional passive eavesdropping without jamming and the proactive eavesdropping with constant-power jamming.

Note that there have been a handful of existing works investigating malicious active eavesdropping attacks in the wireless physical layer security literature, where the attacker can jam and eavesdrop in a half-duplex \cite{AmariucaiWei2012,BasciftciGungorKoksalOzguner2015,MukherjeeSwindlehurst2013} or full-duplex mode \cite{MukherjeeSwindlehurst2011,ZhouMahamHjorungnes2011,KapetanovicZhengRusek2015,XiongLiangLiGong2015}. However, this line of research aimed to design defense methods against {\it illegitimate} active eavesdropping attacks, while in this letter we utilize the {\it legitimate} proactive eavesdropping for the purpose of surveillance, by leveraging jamming in a full-duplex mode.

\section{System Model and Problem Formulation}

As shown in Fig. \ref{fig:0}, we consider a legitimate surveillance scenario, where a full-duplex legitimate monitor aims to eavesdrop a point-to-point suspicious communication link via jamming. The suspicious transmitter and receiver are each equipped with a single antenna, and the legitimate monitor is equipped with two antennas, one for eavesdropping (receiving) and the other for jamming (transmitting). We consider a block-fading frequency-nonselective channel, where the wireless channels remain constant over each transmission block. Over any time block, let $g_0$, $g_1$, and $g_2$ denote the channel power gains from the suspicious transmitter to the suspicious receiver, from the suspicious transmitter to the eavesdropping antenna of the legitimate monitor, and from the jamming antenna of the legitimate monitor to the suspicious receiver, respectively. We consider Rayleigh fading for wireless channels, and thus $g_0$, $g_1$, and $g_2$ are modeled as independent exponentially distributed random variables with parameters $\lambda_0$, $\lambda_1$, and $\lambda_2$, respectively, i.e., $g_i\sim \mathrm{Exp}(\lambda_i)$ with the mean being $1/\lambda_i, i\in\{0,1,2\}$. It is assumed that the legitimate monitor does not know the CSI of $g_0$ or $g_2$ at each time block, while it knows their channel distribution information (CDI). Furthermore, we assume that the channel $g_1$ is perfectly known at the legitimate monitor in each time block.

We consider that the suspicious transmitter transmits with a constant power $P$, while the legitimate monitor employs a jamming power $Q$ to interfere with the suspicious receiver to facilitate the simultaneous eavesdropping. Accordingly, the received signal-to-interference-plus-noise ratio (SINR) at the suspicious receiver and the received signal-to-noise ratio (SNR) at the legitimate monitor are denoted as $\gamma_0 = \frac{g_0P}{g_2Q + \sigma_0^2}$ and $\gamma_1 = \frac{g_1P}{\sigma_1^2}$, where $\sigma_0^2$ and $\sigma_1^2$ denote the noise powers at the suspicious receiver and the legitimate monitor, respectively. Here, the self-interference from the jamming antenna to the eavesdropping antenna at the legitimate monitor is assumed to be perfectly cancelled by using advanced analog and digital self-interference cancellation methods. Assuming that the information signal transmitted by the suspicious transmitter and the jamming signal sent by the legitimate monitor are both Gaussian distributed \cite{KashyapBasarSrikant2004}, the achievable data rates of the suspicious link and the eavesdropping link (in bps/Hz) are respectively expressed as
\begin{align}
r_0 = & \log_2\left(1+\frac{g_0P}{g_2Q + \sigma_0^2}\right), \\
r_1 = & \log_2\left(1+\frac{g_1P}{\sigma_1^2}\right).
\end{align}

Since the suspicious transmitter is not aware of the instantaneous CSI of the suspicious link, it employs a fixed transmission rate, denoted by $R$, over all transmission blocks. In each time block, if the achievable data rate $r_0$ (or $r_1$) is greater than or equal to $R$, then the suspicious receiver (or the legitimate monitor) can correctly decode the data sent by the suspicious transmitter. Otherwise, it cannot decode correctly, and declares an outage. Accordingly, the decoding and eavesdropping outage probabilities at the suspicious receiver and the legitimate monitor are respectively denoted as
\begin{align}
p^{\rm out}_0 &= \mathbb{P}\left(r_0 < R\right), \\
p^{\rm out}_1 &= \mathbb{P}\left(r_1 < R\right). \label{eqn:p1}
\end{align}
As a result, we define the average eavesdropping rate at the legitimate monitor as the suspicious communication rate $R$ multiplied by the non-outage probability $1-p^{\rm out}_1$, i.e., $R^{\rm avg} \triangleq R(1-p^{\rm out}_1)$, which denotes the average rate correctly decoded (eavesdropped) over a long time.

In practice, the suspicious transmitter adjusts its communication rate $R$ to maintain the decoding outage probability $p^{\rm out}_0$ at the suspicious receiver to be fixed at a certain level, i.e., $p^{\rm out}_0 = \delta$, with $\delta > 0$ denoting the target outage probability. By exploiting such a property, the legitimate monitor can adjust its jamming power $Q$ to lessen the achievable data rate $r_0$ of the suspicious link, thus reducing the suspicious communication rate $R$ to keep $p^{\rm out}_0 = \delta$. In particular, we aim to optimize the jamming power $Q$ at the legitimate monitor to maximize its average eavesdropping rate $R^{\rm avg}$, for which the optimization problem is formulated as
\begin{align}
{\rm (P1)}:\max_{Q,R} ~&R(1-p^{\rm out}_1)\nonumber\\
\mathrm{s.t.}~&p^{\rm out}_0 = \delta,\label{eqn:con1}\\
&0\le Q\le Q_{\max},\label{eqn:con2}
\end{align}
where $Q_{\max} > 0$ denotes the maximally allowable jamming power of the legitimate monitor. Note that in problem (P1) the suspicious communication rate $R$ is an auxiliary optimization variable, which changes as a function of $Q$ according to (\ref{eqn:con1}) (as will be shown later). In general, there is a trade-off in determining the optimal jamming power $Q$ to maximize $R^{\rm avg}$ in (P1): a larger $Q$ can reduce $R$ more significantly, which in turn helps improve the non-outage probability $1-p_1^{\rm out}$ at the legitimate monitor. This motivates us to optimize $Q$ to maximize $R^{\rm avg}$.

\begin{figure*}
\begin{minipage}[t]{0.325\linewidth}
\centering
\includegraphics[width=\textwidth]{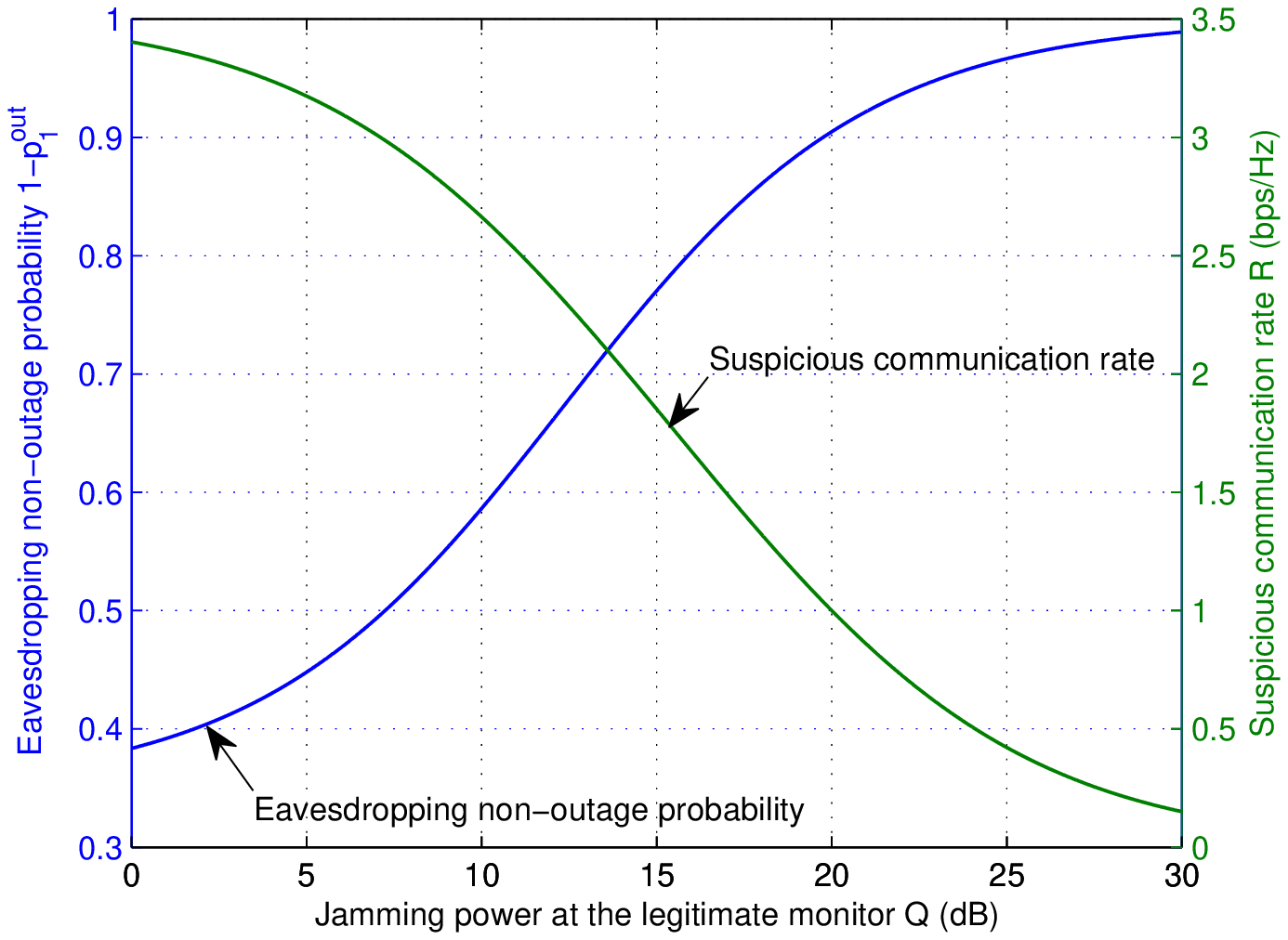}\label{fig:2}
\caption{The eavesdropping non-outage probability $1-p_1^{\rm out}$ and the suspicious communication rate $R$ versus the jamming power at the legitimate monitor $Q$.}\vspace{-0.5em}
\end{minipage}
\hfill
\begin{minipage}[t]{0.325\linewidth}
\centering
\includegraphics[width=\textwidth]{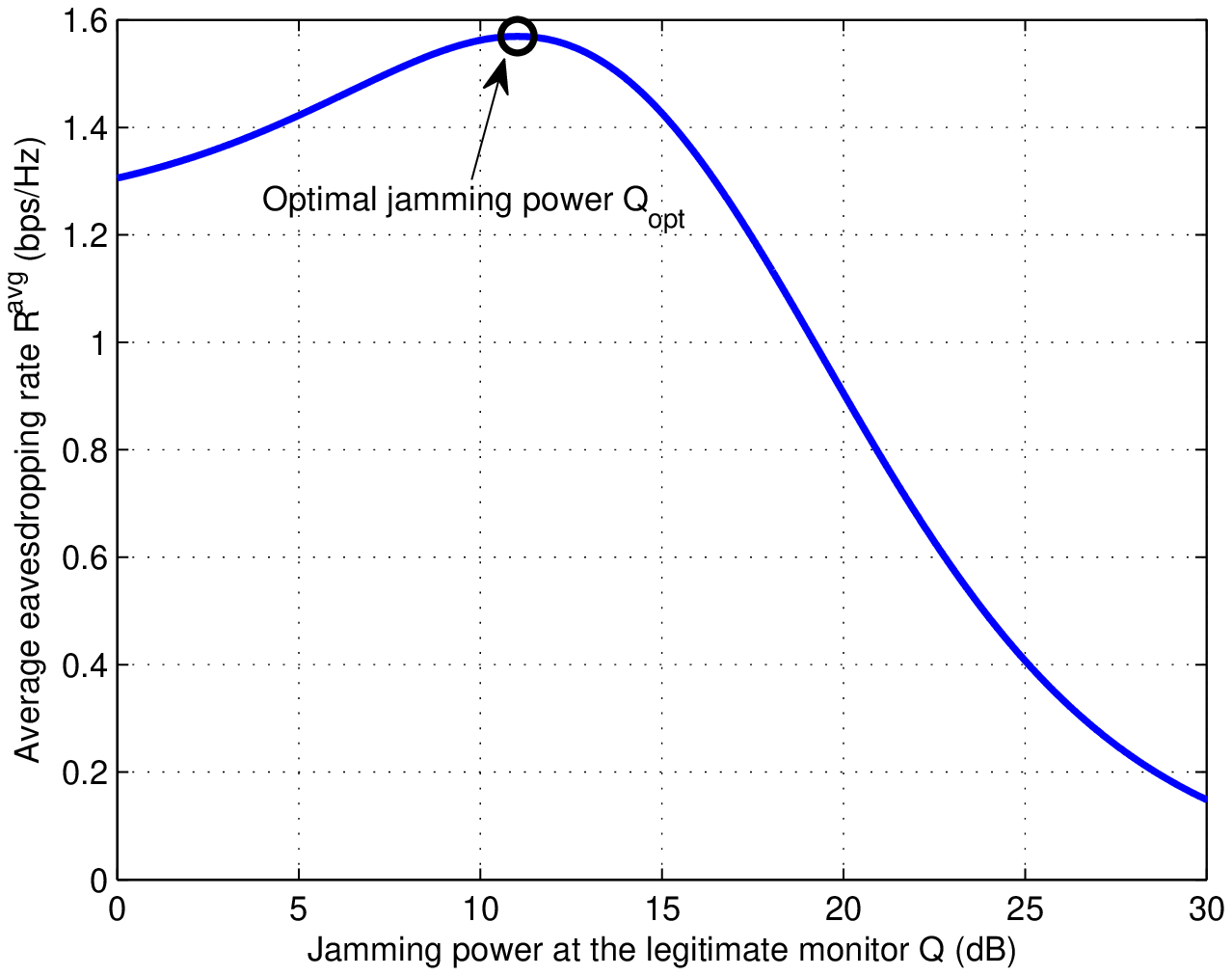}\label{fig:3}
\caption{The average eavesdropping rate $R^{\rm avg}$ versus the jamming power at the legitimate monitor $Q$.}\vspace{-0.5em}
\end{minipage}
\hfill
\begin{minipage}[t]{0.325\linewidth}
\centering
\includegraphics[width=\textwidth]{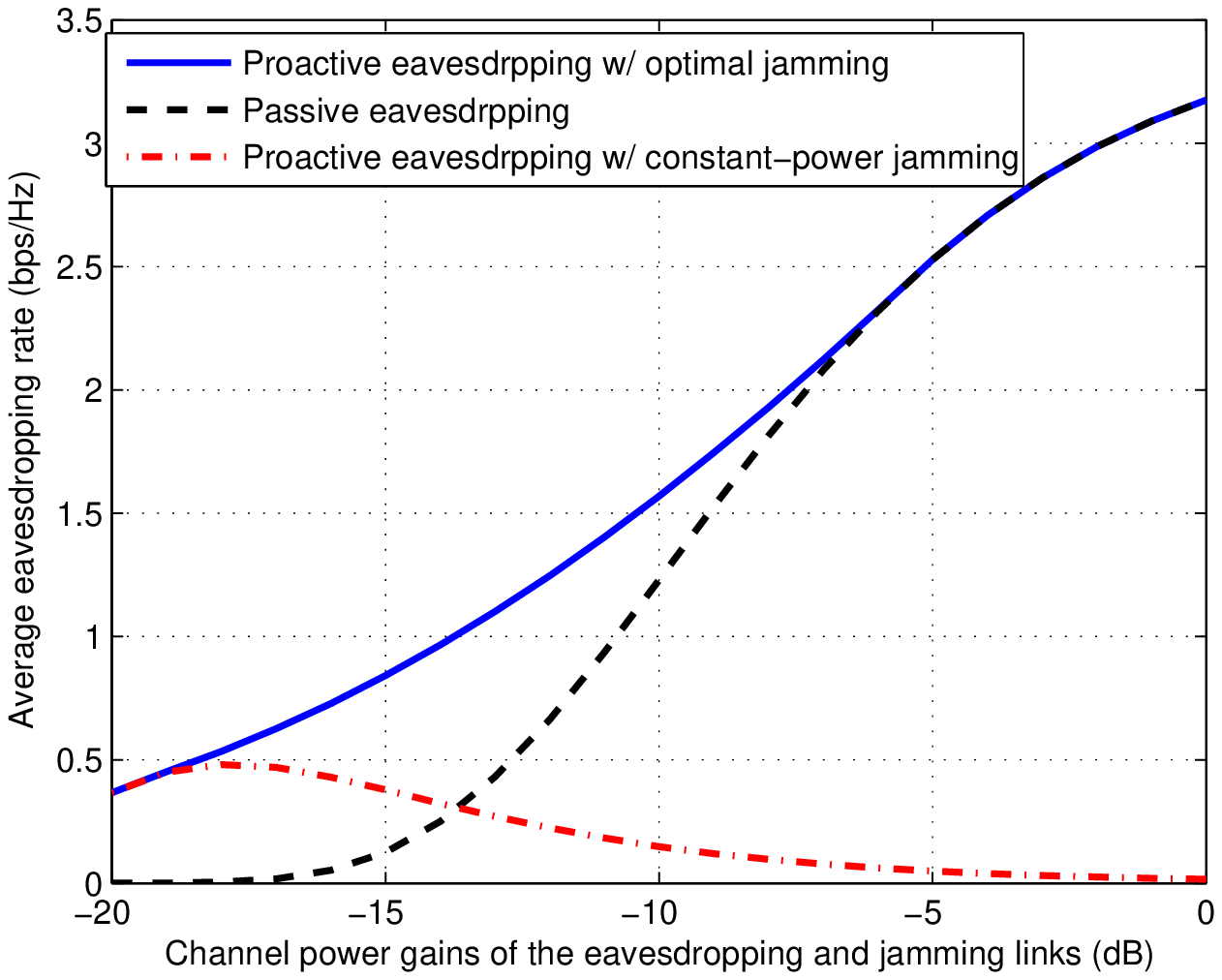}\label{fig:4}
\caption{The average eavesdropping rate $R^{\rm avg}$ versus the average channel power gains of the eavesdropping and jamming links.}\vspace{-0.5em}
\end{minipage}
\end{figure*}

\section{Optimal Jamming for Proactive Eavesdropping}

In this section, we provide the optimal jamming power solution for the legitimate monitor by solving problem (P1).

First, we derive the one-to-one relationship between $R$ and $Q$ and transform problem (P1) into a single-variable optimization problem over $R$. We have the following lemma.
\begin{lemma}\label{proposition:2}
It follows from (\ref{eqn:p1}) that
\begin{align}
p_0^{\rm out} = & 1-\frac{\lambda_2/((2^R - 1)Q)}{\lambda_0/P + \lambda_2/((2^R - 1)Q)}e^{-\lambda_0 \sigma_0^2 (2^R-1)/P}.\label{eqn:p0:all}
\end{align}
By combing this with (\ref{eqn:con1}), we have
\begin{align}\label{eqn:Q:opt}
Q = \psi(R) \triangleq \frac{P\lambda_2e^{-\lambda_0 \sigma_0^2 (2^R-1)/P} }{\lambda_0(2^R - 1)(1-\delta)} - \frac{\lambda_2P}{\lambda_0(2^R - 1)}.
\end{align}
\end{lemma}
\begin{proof}
See Appendix \ref{app:B}.
\end{proof}
Here, the jamming power function $\psi(R)$ in (\ref{eqn:Q:opt}) is monotonically decreasing in $R$. This intuitively means that to reduce the suspicious communication rate $R$, higher jamming power $Q$ is required for the legitimate monitor. By using (\ref{eqn:Q:opt}) together with the constraint in (\ref{eqn:con2}), we have $0 \le \psi(R) \le Q_{\max}$, and accordingly
\begin{align}\label{eqn:R:opt}
\psi^{-1}(Q_{\max})\le R\le \psi^{-1}(0),
\end{align}
where $\psi^{-1}(\cdot)$ is the inverse function of $\psi(\cdot)$. Specifically, we have $\psi^{-1}(0) = \log_2\left(1 + \frac{-P\ln(1-\delta)}{\lambda_0\sigma_0^2}\right)$ and
\begin{align*}
\psi^{-1}(Q) = \log_2\left(1+\frac{P}{\sigma_0^2\lambda_0}\mathcal{W}\left( \frac{\sigma_0^2\lambda_2}{Q(1-\delta)}e^{  \frac{\sigma_0^2 \lambda_2}{Q}}\right) - \frac{\lambda_2P}{Q\lambda_0}\right)
\end{align*}
for any $Q > 0$, where $\mathcal{W}\left(x\right)$ denotes the Lambert W function of $x$ with $\mathcal{W}(x)e^{\mathcal{W}(x)} = x$ \cite{Lambert}. Furthermore, note that
\begin{align}
p^{\rm out}_1 & =  \mathbb{P}\left(g_1 < (2^R - 1) \frac{\sigma_0^2}{P}\right) = 1 - e^{-\lambda_1(2^R - 1) \frac{\sigma_0^2}{P}}
\end{align}
due to the fact that $g_1\sim\mathrm{Exp}(\lambda_1)$. As a result, problem (P1) is equivalently reformulated as the following single-variable optimization problem over $R$:
\begin{align}
{\rm (P2)}:\max_{R} ~& R(1-p^{\rm out}_1) = R e^{-\lambda_1(2^R - 1) \frac{\sigma_0^2}{P}}\nonumber\\
\mathrm{s.t.}~&(\ref{eqn:R:opt}).\nonumber
\end{align}

Next, we study problem (P2), for which we have the following lemma.
\begin{lemma}\label{proposition:1}
The objective function of (P2) is monotonically increasing over $[0, R^*]$, and monotonically decreasing over $(R^*,+\infty)$, where $R^*$ is the suspicious communication rate achieving the maximum average eavesdropping rate, given by
\begin{align}\label{eqn:R_star}
R^*= \frac{1}{\ln 2}\mathcal{W}\left(\frac{P}{\lambda_1\sigma_0^2}\right).
\end{align}
\end{lemma}
\begin{proof}
See Appendix \ref{app:A}.
\end{proof}

From Lemma \ref{proposition:1}, we can easily obtain the optimal solutions to (P2) and (P1) in the following theorem, for which the proof is omitted for brevity.
\begin{theorem}\label{theorem1}
The optimal suspicious communication rate $R$ to (P2) and (P1) and the optimal jamming power $Q$ to (P1) are respectively given as
\begin{align}
R_{\rm opt} &= \max\left(\min\left(\psi^{-1}(0), R^*\right),\psi^{-1}(Q_{\max})\right),\\
Q_{\rm opt} &= \psi(R_{\rm opt})= \min\left(\max\left(0, \psi(R^*)\right),Q_{\max}\right).
\end{align}
\end{theorem}

\begin{remark}\label{remark1}
The above optimal solution to (P1) is intuitively explained as follows. First, consider the case when $R^*$ is no smaller than the suspicious communication rate $\psi^{-1}(0)$ with zero jamming power, which may happen when the eavesdropping channel is similar to or stronger than the suspicious channel. In this case, we have $R_{\rm opt} = \psi^{-1}(0)$ and thus $Q_{\rm opt} = 0$, which means that the no jamming is required for the legitimate monitor to maximize its average eavesdropping rate. Next, consider $R^* < \psi^{-1}(0)$ when the eavesdropping channel is weaker than the suspicious channel. In this case, the legitimate monitor should use a positive jamming power $Q_{\rm opt} = \psi(R^*)$ to interfere with the suspicious receiver (as long as the jamming power does not exceed the maximum value $Q_{\max}$), such that the suspicious transmitter reduces its communication rate from $\psi^{-1}(0)$ to $R^*$ for maintaining the outage probability at the suspicious receiver to be $p_0^{\rm out} = \delta$, thus maximizing the average eavesdropping rate $R^{\rm avg}$.
\end{remark}


\section{Numerical Results}

In this section, we provide numerical results to validate our proposed proactive eavesdropping via jamming approach. In the simulation, we set the average channel power gain of the suspicious link as $1/\lambda_0 = 1$, with the parameter of the exponentially distributed variable $g_0$ being $\lambda_0 = 1$. Here, the channel power gains are normalized over the receiver noise powers such that we can conveniently set the noise powers to be $\sigma_0^2 = \sigma_1^2 = 1$. Furthermore, we set the target outage probability at the suspicious link, the transmit power of the suspicious transmitter, and the maximum jamming power at the legitimate monitor to be $\delta=0.05$, $P = 20$ dB, and $Q_{\max} = 30$ dB, respectively.

Fig. 2 shows the eavesdropping non-outage probability $1-p_1^{\rm out}$ and the suspicious communication rate $R$ versus the jamming power at the legitimate monitor $Q$, where we set the average channel power gains of the eavesdropping and jamming links as $1/\lambda_1 = 1/\lambda_2 = 0.1$, with the parameters $\lambda_1 = \lambda_2 = 10$. This corresponds to the practical case when the legitimate monitor is far away from the suspicious transmitter and receiver. It is observed that as the jamming power $Q$ increases, the suspicious communication rate $R$ decreases, while the eavesdropping non-outage probability $1-p_1^{\rm out}$ increases. Fig. 3 shows the average eavesdropping rate $R^{\rm avg} = R(1-p_1^{\rm out})$ versus the jamming power at the legitimate monitor $Q$. It is observed that as $Q$ increases,  $R^{\rm avg}$ first increases and then decreases. This is consistent with Lemma \ref{proposition:1}. The maximum average eavesdropping rate is observed to be achieved at the optimal jamming power $Q_{\rm opt}$ obtained in Theorem \ref{theorem1}.

%

Fig. 4 shows the average eavesdropping rate $R^{\rm avg}$ versus the average channel power gains of the eavesdropping and jamming links, i.e., $1/\lambda_1 = 1/\lambda_2$. We consider the following two reference schemes for performance comparison.
\begin{itemize}
  \item {\it Conventional passive eavesdropping without jamming:} in this scheme, the legitimate monitor does not send any jamming signal, i.e., $Q = 0$.
  \item {\it Proactive eavesdropping with constant-power jamming:} in this scheme, the legitimate monitor uses the maximum jamming power, i.e., $Q = Q_{\max}$.
\end{itemize}
It is observed that as the average channel gains of the eavesdropping and jamming links increase, the average eavesdropping rates of the passive eavesdropping and the proactive eavesdropping with optimal jamming  increase, while that of the proactive eavesdropping with constant-power jamming first increases and then decreases. The proactive eavesdropping with optimal jamming is observed to outperform the two reference schemes. Specifically, when the channel gains are small (e.g., smaller than $-15$ dB), the average eavesdropping rate of passive eavesdropping is observed to be close to zero, since the legitimate monitor is too far from the suspicious transmitter and thus is difficult to eavesdrop successfully. In contrast, the proactive eavesdropping (with both optimal and constant-power jamming) is observed to achieve significant average eavesdropping rate gain in this case. When the channel gains become large (e.g., larger than $-5$ dB), the passive eavesdropping and the proactive eavesdropping with optimal jamming schemes are observed to achieve identical average eavesdropping rates. This is consistent with our observations in Remark \ref{remark1}.


\section{Conclusion}

This letter proposes a new legitimate eavesdropping paradigm in wireless security, and presents a proactive eavesdropping via jamming approach, where the legitimate monitor jams the suspicious wireless communication link for improving the eavesdropping performance. By considering Rayleigh fading channels, we obtain the optimal jamming power in closed-form to maximize the average eavesdropping rate. Thanks to the jamming with optimized power control, our proposed proactive eavesdropping improves the average eavesdropping rate significantly over the conventional passive eavesdropping scheme without jamming, especially when the legitimate monitor is far from the suspicious transmitter and receiver. We believe that this work can provide new insights on wireless legitimate surveillance for the national security.

\appendix

\subsection{Proof of Lemma \ref{proposition:2}}\label{app:B}
From (\ref{eqn:p1}), it follows that $p^{\rm out}_0$ can be re-expressed as
\begin{align}
p^{\rm out}_0
& = \mathbb{P}\left({g_0P} - g_2(2^R - 1)Q < \sigma_0^2(2^R - 1) \right).\label{eqn:p0out}
\end{align}
First, we obtain the probability density function (PDF) of $Z={g_0P} - (2^R - 1)g_2Q = X_1-X_2$, where $X_1 \triangleq g_0P$ and $X_2 \triangleq g_2(2^R - 1)Q$. Note that $g_i \sim {\rm Exp}(\lambda_i), i\in\{0,2\}$. Accordingly, denote $\tilde\lambda_1 = \lambda_0/P$ and $\tilde\lambda_2 = \lambda_2/((2^R - 1)Q)$. Then we have $X_i \sim {\rm Exp}(\tilde\lambda_i), i\in\{1,2\}$. 
Therefore, the PDF of $Z=X_1-X_2$ is expressed as
\begin{align*}
f_Z(z) 
=& \int_{0}^{+\infty}f_{X_2}(z+x_2)\tilde\lambda_2 e^{-\tilde\lambda_2x_2}{\rm d} x_2 \\
=& \left\{\begin{array}{ll}
\frac{\tilde\lambda_1\tilde\lambda_2}{\tilde\lambda_1+\tilde\lambda_2}e^{-\tilde\lambda_1z}, & {\rm if}~z > 0\\
\frac{\tilde\lambda_1\tilde\lambda_2}{\tilde\lambda_1+\tilde\lambda_2}e^{\tilde\lambda_2z}, & {\rm if}~z \le  0.
\end{array}
\right.
\end{align*}
As a result, it follows that when $z \ge 0$,
\begin{align}
\mathbb{P}(Z<z)
=& 1 - \frac{\tilde\lambda_2}{\tilde\lambda_1+\tilde\lambda_2} e^{-\tilde\lambda_1 z}.\label{eqn:f_z}
\end{align}
By combining (\ref{eqn:p0out}) and (\ref{eqn:f_z}), we have $p_0^{\rm out} = \mathbb{P}(Z<\sigma_0^2 (2^R-1))$ and thus (\ref{eqn:p0:all}) follows. Using (\ref{eqn:p0:all}) together with $p^{\rm out}_0 = \delta$ in (\ref{eqn:con1}), we have (\ref{eqn:Q:opt}). Therefore, this lemma is proved.


\subsection{Proof of Lemma \ref{proposition:1}}\label{app:A}
By letting $x = 2^R-1\ge 0$ and accordingly $R = \log_2(1+x)$, we can rewrite the objective function of (P2) as
\begin{align}
\phi(x) = \log_2(1+x) e^{-\lambda_1x \frac{\sigma_0^2}{P}}.
\end{align}
As a result, proving this lemma is equivalent to showing that $\phi(x)$ is monotonically increasing over $[0, x^*]$, and monotonically decreasing over $(x^*,+\infty)$, where $x^* = 2^{R^*}-1$.

Taking the first-order derivative of $\phi(x)$, we have
\begin{align*}
\phi'(x) 
= &e^{-\frac{\lambda_1\sigma_0^2x}{P}}\left(\frac{1}{\ln 2\cdot(1+x)} -  \frac{\lambda_1\sigma_0^2}{P}\log_2(1+x)\right).
\end{align*}
By setting $\phi'(x) = 0$, we have $x^* = 2^{R^*}-1$ with $R^*$ given in (\ref{eqn:R_star}). Furthermore, it is easy to show that $\phi'(x) > 0$ when $x \in [0, x^*]$ and $\phi'(x) < 0$ when $x \in (x^*,\infty)$. Therefore, this lemma is proved.
%
%

\end{document}